\documentclass[conference,letterpaper]{IEEEtran}
\usepackage[utf8]{inputenc} 
\usepackage[T1]{fontenc}
\usepackage{url}
\usepackage{ifthen}
\usepackage{cite}
\usepackage{amsmath,amssymb,amsfonts}
\usepackage{amsthm}
\usepackage[inline]{enumitem}
\usepackage{textcomp}
\usepackage{xcolor}
\usepackage{float}
\def\BibTeX{{\rm B\kern-.05em{\sc i\kern-.025em b}\kern-.08em
    T\kern-.1667em\lower.7ex\hbox{E}\kern-.125emX}}

\newtheorem{lemma}{Lemma}
\newtheorem{remark}{Remark}

\newtheorem{auxiliary code}{Auxiliary Code}

\usepackage{algorithm}
\usepackage[noend]{algpseudocode}
\algrenewcommand\algorithmicindent{0.8em}%
\usepackage{tikz}
\usepackage{graphicx}
\usetikzlibrary{positioning}
\usepackage[font={small}]{caption}
\usepackage{courier}
\usepackage{subcaption}
\usepackage{multirow}
\usepackage[T1]{fontenc}    
  \usepackage{array}  

\newcolumntype{L}{>{\centering\arraybackslash}m{2.5cm}}    
\newcolumntype{M}{>{\centering\arraybackslash}m{4cm}} 
\newcolumntype{N}{>{\centering\arraybackslash}m{5cm}} 
\newcolumntype{O}{>{\centering\arraybackslash}m{5cm}} 
    
\DeclareCaptionLabelFormat{andtable}{\tablename~\thetable  \hspace{0.1cm}\& #1~#2   }    
    
\def \codename {SEF } %
\def \tT {\mathcal{T} }
\def \tsymbol {\tau}
\def \tproof {{\fontfamily{qcr}\selectfont Proof}}
\def \troot {{\fontfamily{qcr}\selectfont Root}}
\def \tverify {{\fontfamily{qcr}\selectfont Verify-Inclusion}}

\def \talpha {\alpha_{\min}}

\def \tP {\mathbf{F}}
\def \ttT {\mathbf{T}}
\def \tFG {\mathcal{G}}

\def \tI {\mathcal{A}}
\def \tF {\mathcal{F}}

\def \ttN {\overline{C}}

\def \tSS {\Psi}
\def \tssingle {\psi}
\def \tssingletree {ST}
\def \tLS {{\fontfamily{qcr}\selectfont Leaf-Set}}

\def \tlast {\mu}
\def \tlastVN {\mathcal{V}}

\def \tttN {C}
\def \tttC {\text{{\fontfamily{qcr}\selectfont CodeSym}}}
\def \tttS {\text{{\fontfamily{qcr}\selectfont data}}}
\def \tttP {\text{{\fontfamily{qcr}\selectfont parity}}}
\def \tparity {{\fontfamily{qcr}\selectfont encodeParity}}

\def \tparent {{\fontfamily{qcr}\selectfont formParent}}

\def \tdecodelayer {\text{{\fontfamily{qcr}\selectfont decodeLayer}}}

\def \tdropped {\text{{\fontfamily{qcr}\selectfont dropped}}}

\def \tfrozen {\text{{\fontfamily{qcr}\selectfont frozen}}}

\def \tIm {\widehat{\mathcal{A}}}
\def \tFm {\widehat{\mathcal{F}}} 
\def \tNm {\widehat{N}}  
\def \tkm {\widehat{K}} 
 
\def\closed#1#2#3{#1 \in [#2,#3]}
\def\openleft#1#2#3{#1 \in (#2,#3]}
\def\closedN#1#2{#1 \in [#2]}

\newcommand\deb[1]{{\color{black}#1}}  \newcommand\lev[1]{{\color{black}#1}}  

\usepackage{accents}
\newcommand\thickbar[1]{\accentset{\rule{.3em}{.6pt}}{#1}}

\urldef{\bitcoinsize}\url{https://www.blockchain.com/charts/avg-block-size}

\urldef{\bitcoincash}\url{https://bitinfocharts.com/comparison/bitcoin%20cash-size.html#3y}
	
\urldef{\bitcoinsv}\url{https://bitinfocharts.com/comparison/bitcoin%20sv-size.html#3y}

\newcommand\revision[1]{{\color{black}#1}}
\newcommand\trim[1]{{\color{black}#1}}

\setlength{\textfloatsep}{0pt}%

\begin{document}

\title{Polar Coded Merkle Tree: Improved Detection of Data Availability Attacks in Blockchain Systems\vspace{-0.4cm}}

\author{\IEEEauthorblockN{Debarnab Mitra, Lev Tauz, and Lara Dolecek}
\IEEEauthorblockA{Department of Electrical and Computer Engineering, University of California, Los Angeles, USA\\
email: debarnabucla@ucla.edu, levtauz@ucla.edu, dolecek@ee.ucla.edu}
\vspace{-1.12cm}}

\maketitle

\begin{abstract}
    \deb{
    \trim{Light nodes in blockchain systems are known to be vulnerable to \textit{data availability} (DA) attacks} where they accept an invalid block with unavailable portions.  
    \revision{Previous works have used LDPC and 2-D Reed Solomon (2D-RS) codes with Merkle Trees to mitigate DA attacks. While these codes have demonstrated improved performance across a variety of metrics such as DA detection probability, they are difficult to apply to blockchains with large blocks due to generally intractable code guarantees for large codelengths (LDPC), large decoding complexity (2D-RS), or large coding fraud proof sizes (2D-RS).}
    We address these issues by proposing the novel \emph{Polar Coded Merkle Tree} (PCMT) which is a Merkle Tree built from the encoding graphs of polar codes and a specialized polar code construction called \emph{Sampling-Efficient Freezing} (SEF). We demonstrate that the PCMT with SEF polar codes performs well in detecting DA attacks for large block sizes.
    }

\end{abstract}

\vspace*{-0.3cm}
\section{Introduction}
\vspace{-0.18cm}
Decentralization and security properties of blockchains have led to their applications in
\trim{a wide variety of fields \cite{Bitcoin,ethereum,AI-1,AI-2,supplychain,IOT,healthcare}.}
A blockchain is an immutable ledger of transaction blocks. 
\textit{Full nodes} in a blockchain system store the entire ledger and validate transactions. However, for better scalability, blockchains also run \emph{light nodes} \trim{who} store the header 
of each block %
and cannot validate transactions. Light nodes rely on honest full nodes for \emph{fraud proofs} \cite{dataAvailOrg} in order to reject invalid blocks. 

Blockchains where light nodes are connected to a majority of malicious full nodes are vulnerable to \emph{data availability (DA) attacks} \cite{dataAvailOrg, CMT}. In this attack, a malicious full node \trim{(i.e., an adversary)} generates a block with invalid transactions and hides the invalid portion of the block from other  nodes. This action prevents honest nodes from sending fraud proofs to the light nodes. Light nodes, in this scenario, randomly request/sample chunks of the block from \trim{the block generator and detect} a DA attack \trim{if any request} is rejected. 
\trim{To improve the detection of a DA attack by light nodes, erasure coding has been proposed to encode the block \cite{dataAvailOrg} which forces the adversary to hide a larger fraction of the encoded block.}
However, erasure coding allows the adversary to carry out an \emph{incorrect-coding (IC) attack} by incorrectly generating the coded block, in which case honest full nodes can send an \emph{IC proof} to the light nodes to reject the block \cite{dataAvailOrg}, \cite{CMT}.
 \revision{Recently, \cite{InformationDispersal} proposed a technique to mitigate DA attacks without requiring IC-proofs, however \cite{InformationDispersal} employs complex cryptographic computations.}
We remark that channel coding has been considered to mitigate  a variety of issues in blockchains
\cite{SSskewITW,SSskew-full,cover,erasurelowstorage,networkcodingstorage,downsampling,patternederasure,SeF,secureraptor,secregenerating,networkcodedPBFT,polyshard,crossshard,AceD,DE-PEG}.

\deb{ 
The first paper to address DA attacks  via coding used 2D Reed-Solomon (RS) codes to encode each block \cite{dataAvailOrg}.  \revision{Their method was recently optimized in \cite{RSoptimize}}. While offering a high probability of detecting DA attacks, 2D-RS codes result in large IC proof sizes and decoding complexity. \revision{Large IC-proof sizes can be exploited by the adversary to congest the network and cause issues such as denial of service attacks, reduced transaction throughput, etc..}
To improve upon \revision{2D-RS codes}, authors in \cite{CMT} proposed the \emph{Coded Merkle Tree (CMT)}: a Merkle tree \cite{Bitcoin} where each layer is encoded using a Low-Density Parity-Check (LDPC) code. LDPC codes reduce the IC proof size compared to 2D-RS codes due to their sparse parity check equations.  LDPC codes also allow the use of a low complexity peeling decoder \cite{ModernCodingTheory} for decoding the CMT where the probability of failure to detect DA attacks depends on the minimum stopping set size of the LDPC code \cite{CMT}.}

\lev{Application of coding has to be carefully considered depending on the size of the transaction blocks in blockchains which can range from a few MBs (small block size), e.g., Bitcoin \cite{BitcoinSize}, Bitcoin Cash\cite{BitcoinCash}, to hundreds of MBs (large block size), e.g., Bitcoin SV\cite{BitcoinSV}. For large block sizes, large code lengths are required since large code lengths allow for smaller partitioning of the block, thereby reducing the network bandwidth requirement.
In regard to the CMT, previous work in \cite{CMT} has considered random LDPC codes for large code lengths. However, due to the random construction, the approach has a non-negligible probability of generating bad codes \cite{CMT} which undermines the security of the system. At the same time, works in \cite{SSskewITW,SSskew-full} have provided deterministic LDPC codes with short code lengths based on the PEG algorithm \cite{PEG} that result in a low probability of failure. 
However, the NP-hardness of determining the minimum stopping set size of  LDPC codes \cite{SSNP-hard} makes it difficult to provide an efficiently computable guarantee on the probability of failure for large code lengths. 
As such, designing a CMT with LDPC codes at large code lengths/block sizes is difficult.
To mitigate these issues, we propose \emph{Polar Coded Merkle Tree (PCMT)}: a CMT built using the encoding graph of polar codes \cite{PolarCodesErikan} (we refer to a CMT built using LDPC codes as an LCMT). Although polar codes have dense parity check matrices \cite{polarhighdensity}, they have sparse encoding graphs, which  result in a small IC proof size in the PCMT. We provide a specialized  polar code construction called \emph{Sampling-Efficient Freezing (SEF) Algorithm} that i) provides a low probability of failure to detect DA attacks;  ii) allows flexibility in designing polar codes of any lengths. 
We demonstrate that \codename Polar codes have an efficiently computable
 guarantee on the probability of failure which simplifies system design at large block sizes. We then demonstrate that for large block sizes, a PCMT built using \codename Polar codes result in a lower probability of failure compared to LDPC codes designed by the PEG algorithm. }

The rest of this paper is organized as follows. In Section \ref{sec:prelims}, we provide the preliminaries and the system model. In Section \ref{sec:PCMT_construction}, we provide our construction method for the PCMT. We present the SEF algorithm in Section \ref{sec:SEF_algo}. Finally, we provide \deb{simulation results and performance comparisons in Section \ref{sec:sims}.}%

\vspace{-0.25cm}
\section{Preliminaries and System Model}\label{sec:prelims}

\textbf{Notation:} Let $\tP_2 = \begin{bmatrix}
1 & 0 \\
1 & 1 
\end{bmatrix}$ and $\ttT_2 = \begin{bmatrix}
1 \\
2 
\end{bmatrix}$. For a vector $\mathbf{a}$, let $\mathbf{a}(i)$  and $\min(\mathbf{a};k)$ denote the $i$th element and $k$th smallest value of $\mathbf{a}$, respectively.  Let $\otimes n$ denote the $n$th
Kronecker power. 
Let $\vert S \vert$ be the cardinality of set $S$. All logarithms are with base 2. For integers $a$ and $b$ \trim{define $[a,b] = \{i \;\vert\; a \leq i \leq b\}$, $(a,b] = \{i \;\vert\; a < i \leq b\}$, and $[a] = \{i \;\vert\; 1 \leq i \leq a\}$, where elements in the three sets are integers.}

\subsubsection{Coded Merke Tree preliminaries}
A CMT, like a regular Merkle Tree \cite{Bitcoin}, is a cryptographic commitment generator that
is used to check the integrity of transactions in the block \cite{CMT}. Additionally, erasure coding allows to check for data availability 
via random sampling. In this section, we provide a general framework for the CMT construction that captures its key properties. Later in Section \ref{sec:PCMT_construction}, we present the PCMT within the general CMT framework.

A CMT parametrized by $\tT = (k, R, q , l)$ is a coded version of a Merkle tree \cite{Bitcoin}. It has $(l+1)$ layers $L_0,L_1, \ldots, L_l$ where $L_l$ is the base layer and $L_0$ is the CMT root. $L_j$ has $N_j  = \frac{N_l}{(qR)^{l-j}}$ coded symbols, where 
$N_l = \frac{k}{R}$. 
Let $\tttC[j] = \{\tttN_j[i] \;\vert\; \closedN{i}{N_j}\}$ be the coded symbols of $L_j$, where $\tttS[j] = \{\tttN_j[i] \; \vert\; \closedN{i}{RN_j}\}$ and $\tttP[j] = \{\tttN_j[i] \;\vert\; \openleft{i}{RN_j}{N_j}\}$ are the set of data and parity symbols of $L_j$, respectively. The coded symbols in the CMT are formed as follows:
  Set the data symbols of the base layer $\tttS[l]$ to the chunks of the transaction block. For $j = l, l-1, \ldots, 1$: \emph{a)} form the parity symbols $\tttP[j]$ from the data symbols $\tttS[j]$ using a rate $R$ systematic linear code  via a procedure $\tttP[j] = \text{\tparity}(\tttS[j])$; \emph{b)} form the data symbols $\tttS[j-1]$ from the coded symbols $\tttC[j]$ by a procedure  $\tttS[j-1] = \text{\tparent}(\tttC[j])$. 
 
The \tparent$()$ procedure has the property that each data symbol $\tttN_{j-1}[i]$ in $\tttS[j-1]$ contains the hashes of $q$ coded symbols of $\tttC[j]$.  Finally, $\tttS[0]$ forms the root {\troot } of the CMT and is the commitment to the block.

Each coded symbol $\tsymbol$ in the CMT has a Merkle proof \tproof($\tsymbol$) which can be used to check the integrity of the symbol given the {\troot} using \tverify($\tsymbol$,  \tproof($\tsymbol$), \troot).
The CMT is decoded using a hash-aware decoder which decodes the tree from the root to the base layer. Each layer is decoded using a procedure $\tdecodelayer(L_j)$ and the hash of the decoded symbols are matched with their hash provided in the parent layer of the CMT. Assume that the decoded CMT symbols  $\tsymbol_1, \tsymbol_2, \ldots, \tsymbol_{d}$ satisfy a degree $d$ parity check equation (of the erasure code used for encoding). Of these symbols, if there exists a symbol $\tsymbol_e$ whose hash does not match with the hash provided by the parent of $\tsymbol_e$ in the CMT, an IC attack is detected. In this case, an IC proof consists of the following:  the symbols $\{\tsymbol_1, \tsymbol_2, \ldots, \tsymbol_{d}\} \setminus\tsymbol_e $ and their Merkle proofs, and the Merkle proof of $\tsymbol_e$. 
 The IC proof is verified by first verifying that each symbol $\tsymbol_i$, $1 \leq i \leq d, i \neq e$, satisfies \tverify($\tsymbol_i$,  \tproof($\tsymbol_i$), \troot), then decoding ${\tsymbol}_e$ from the remaining symbols (as they form a parity check equation) and then checking that  ${\tsymbol}_e$ does not satisfy \tverify(${\tsymbol}_e$,  \tproof(${\tsymbol}_e$), \troot).

We consider a blockchain system similar to \cite{CMT,SSskew-full} with full nodes and light nodes where full nodes produce new blocks. Light nodes only store the CMT root of each block and use it to verify the Merkle proof of CMT symbols and IC proofs. Similar to \cite{dataAvailOrg,CMT,SSskewITW,SSskew-full}, \trim{we assume that light nodes are honest, are connected to at least one honest full node, but can be connected to a majority of malicious full nodes.} For the purposes of a DA attack, consider one layer of the CMT  having $N$ coded symbols. A malicious full node causes a DA attack by \emph{a)} generating an invalid block and producing its CMT; \emph{b)} hiding coded symbols (of the $N$ coded symbols) such that no honest full node is able to decode back all the coded symbols. Light nodes detect this DA attack by anonymously  and randomly requesting (sampling) a small number of coded symbols from the block producer and accepting  the block if all the requested samples are returned. 
A malicious node only returns coded symbols that it has not hidden \cite{dataAvailOrg,CMT,SSskewITW,SSskew-full}. Let $\talpha$, which we call the \emph{undecodable threshold}, be the minimum number of coded symbols that a malicious node must hide to prevent honest full nodes from decoding all the coded symbols. Then, the probability of failure for \revision{a light node} to detect a DA attack using $s$ random i.i.d. samples is $P_f(s) = (1 - \frac{\talpha}{N})^s$. \revision{Note that we focus on the security of the system on a per client basis similar to \cite{dataAvailOrg},\cite{CMT}.}
The following metrics are of importance for a CMT: 
i) IC proof size (must be small in comparison to the original block size \revision{since this proof is communicated to all light nodes and can be used to congest the network}), ii) undecodable threshold, iii) complexity of computing the undecodable threshold (\deb{which is important at large CMT code length $N$}), and iv) decoding complexity.
In this paper, we will demonstrate a construction of CMT using polar codes called the PCMT that performs well on all these metrics \deb{when the size of the block $b$ is large}.

\begin{figure}[t]
    \centering
    \begin{subfigure}{0.5\linewidth}
\begin{minipage}{0.99\linewidth}
 \centering
     \includegraphics[scale=0.25]{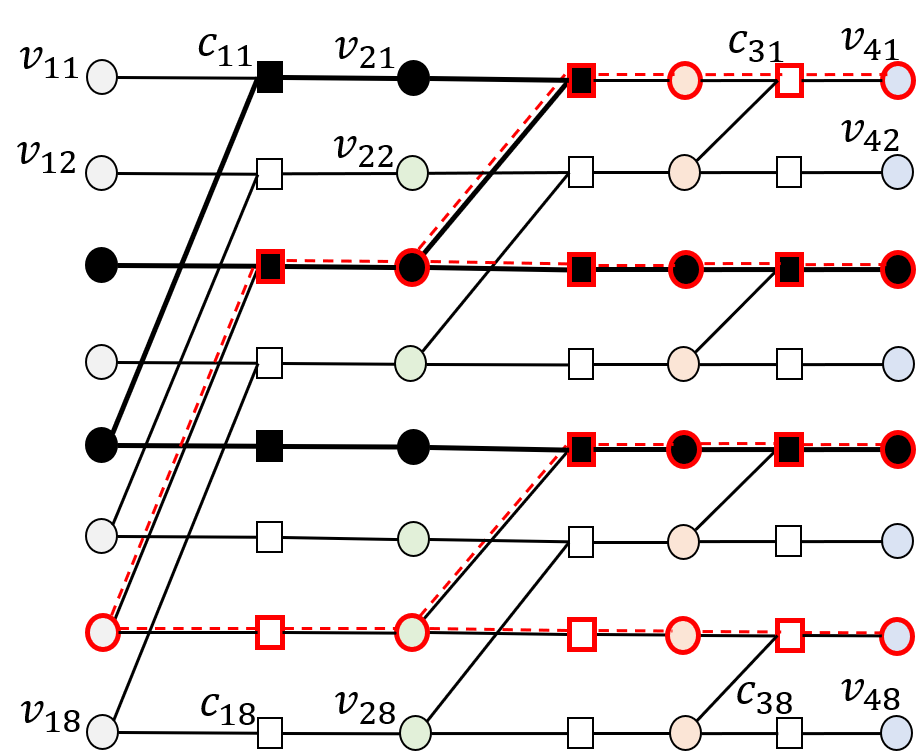}
 \end{minipage}
    \end{subfigure}%
    \begin{subfigure}{0.5\linewidth}
\begin{minipage}{0.99\linewidth}
\centering
     \includegraphics[scale=0.25]{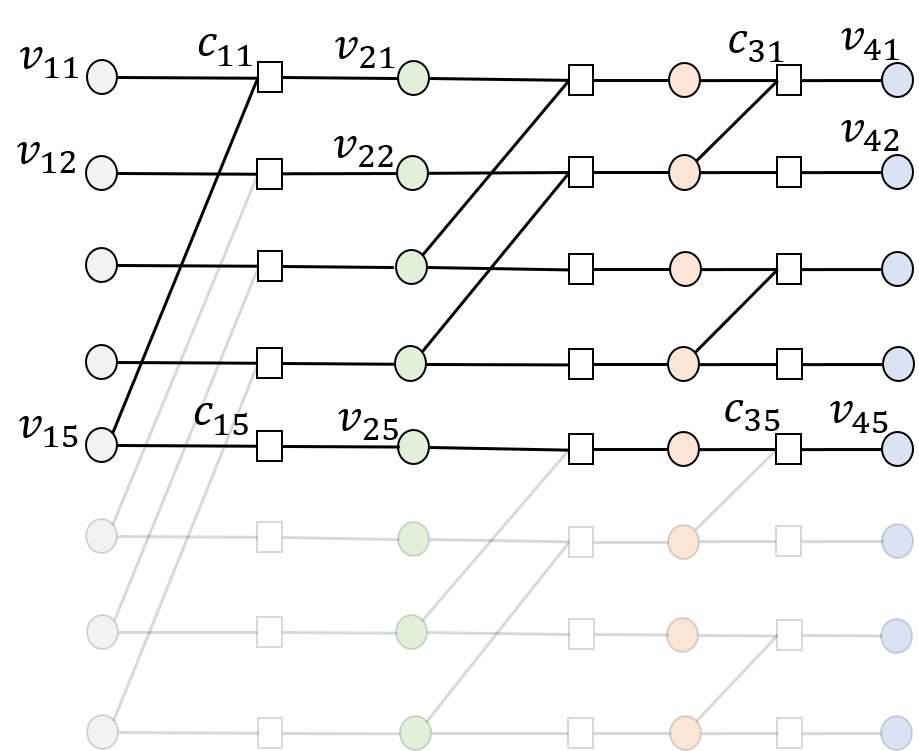}
 \end{minipage}
    \end{subfigure}
    \caption{\footnotesize Left panel: FG $\tFG_{8}$ where circles represent VNs and squares represent CNs.  The black (red) VNs and CNs represent a stopping set (stopping tree); Right panel: $\tFG_5$ obtained by removing the VNs from the last 3 rows of $\tFG_8$. 
   (removed VNs are shown in low opacity). 
    }
    \label{fig:Factor_graph}
\end{figure}

\subsubsection{Polar Codes preliminaries}
An ($\tNm,\tkm$) polar code of codelength $\tNm = 2^n$ for some integer $n$ and information length $\tkm$ is defined by a transformation matrix $\tP_{2^n} = \tP^{\otimes n}_2$.
The generator matrix of the  polar code is a submatrix of $\tP_{2^n}$ having $\tkm$ of its rows corresponding to the data (information) symbols, while the rest of the rows correspond to frozen symbols (zero chunks in this paper).  The \revision{factor graph} (FG) representation \cite{Polar-SStree-TCOM}  of $\tP_{2^3}$ is shown in Fig. \ref{fig:Factor_graph} left panel. In general, the FG of $\tP_{2^n}$ (denoted by $\tFG_{2^n}$) has $n+1$ variable node (VN) and $n$ check node (CN) columns. Let $v_{ki}$ ($c_{ki}$) denote the VN (CN) in the $k$th column and $i$th row as shown in Fig. \ref{fig:Factor_graph}. Note that the CNs have a degree of either 2 or 3.

\begin{figure}[t]
    \centering
    \includegraphics[scale=0.35]{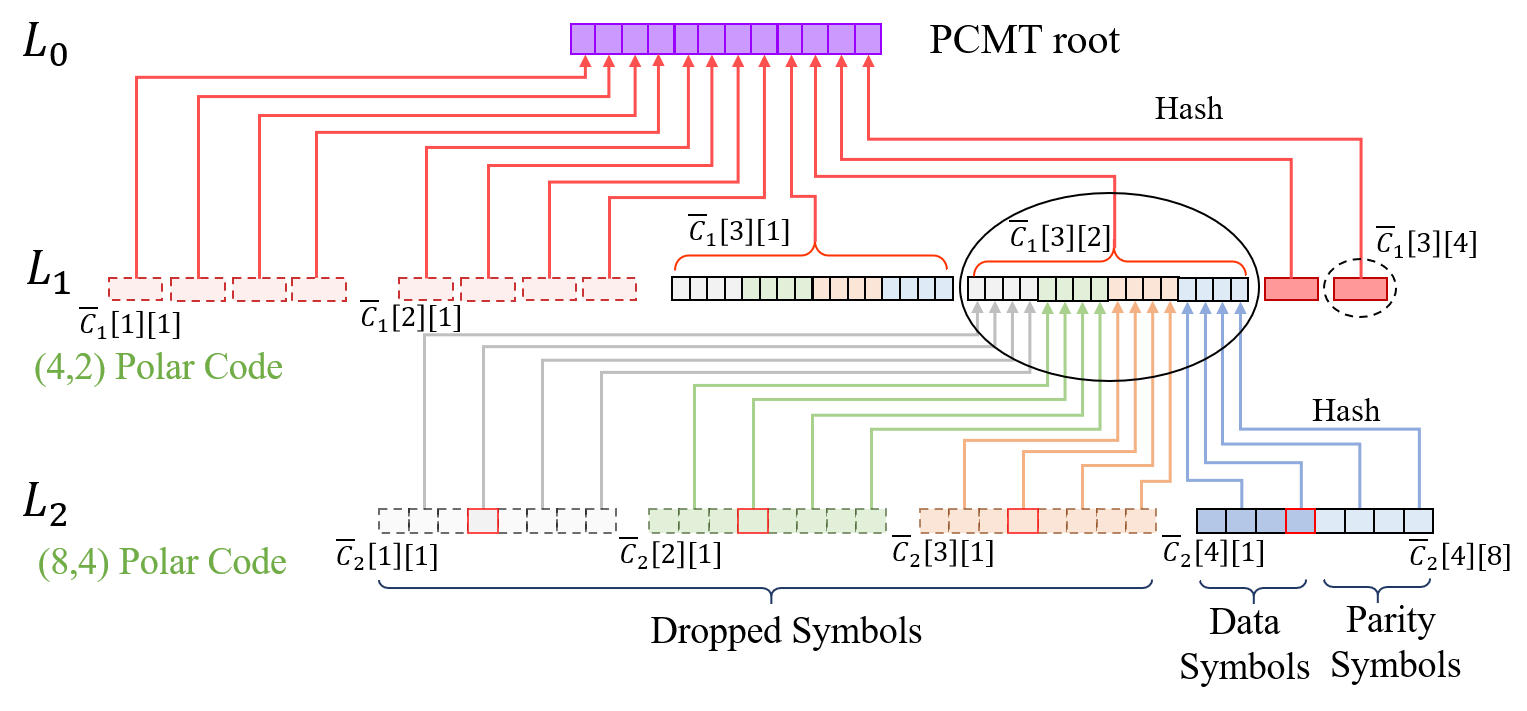}
   \vspace{-20pt}
  \caption{ \footnotesize PCMT $\tT$
$= (k = 4, R = 0.5, q = 4 , l = 2)$. %
In the PCMT, the coded symbols corresponding to all the columns of the polar FG are hashed into the parent layer. 
The dropped symbols are shown in dotted. The symbols in $L_2$ are colored according to the column they belong to in FG $\tFG_8$. The circled symbols in $L_1$ are the Merkle proof of the red symbols in $L_2$. The data (parity) symbols in the Merkle proofs are shown in solid (dashed) circles. }
    \label{fig:PCMT_example}
\end{figure}

Systematic encoding of $(\tNm, \tkm)$ polar codes \cite{SystematicPOlarCodes,SystematicPOlarCodestwostatge}
(required for CMT construction), can be performed using the method described in \cite{SystematicPOlarCodestwostatge}. 
Given information and frozen index sets $\tIm \subset [\tNm]$ and $\tFm = [\tNm] \setminus \tIm$, such that $\vert \tIm \vert = \tkm$, the systematic encoder of \cite{SystematicPOlarCodestwostatge} determines the value of all the VNs in the FG $\tFG_{\tNm}$ such that
i) $\{v_{1i} \;\vert i \;\in \tFm\}$ are set to zero (frozen) symbols;
ii) VNs in $\{v_{(n+1)i} \;\vert\; i \in \tIm\}$ and $\{v_{(n+1)i} \;\vert\; i \in \tFm\}$ are the provided data symbols and resultant parity symbols in the systematic encoding. 
 Details regarding systematic encoding can be found in Appendix \ref{appendix:PDE}.

A polar code can be decoded using a peeling decoder on the code FG. Similar to LDPC codes, the peeling decoder on the FG of a polar code fails if all VNs corresponding to a stopping set (of the FG) are erased. A stopping set is a set of VNs such that every CN connected to this set is connected to at least two VNs in the set. Similar to \cite{Polar-SStree-TCOM}, we call the VNs of a stopping set $\tssingle$ that are in the rightmost column of the FG as its \emph{leaf set} denoted by \tLS($\tssingle$).
An important category of stopping sets in the FG of polar codes is called \emph{stopping trees} \cite{Polar-SStree-TCOM}. A stopping tree is a stopping set that contains only one VN from the leftmost column of the FG (called its root). Fig. \ref{fig:Factor_graph} left panel shows a general stopping set and a stopping tree in the FG $\tFG_{8}$.  As demonstrated in \cite{Polar-SStree-TCOM}, each VN $v_{1i}$, $\closedN{i}{\tNm}$, is root of an unique stopping tree. Let $\tssingletree^{2^n}_i$ be the unique stopping tree with root VN $v_{1i}$ in the FG $\tFG_{2^n}$ and let $f^{2^n}_i = \vert \text{\tLS}(\tssingletree^{2^n}_i) \vert$. It is easy to see (and also proved in \cite{Polar-SStree-TCOM}) that 
$f^{2^n}_i = \ttT_{2^n}(i)$, $\forall \closedN{i}{\tNm}$, where $\ttT_{2^n} = \ttT^{\otimes n}_2$.

\vspace{-0.2cm}
\section{Polar Coded Merkle Tree (PCMT)}\label{sec:PCMT_construction}
\vspace{-0.1cm}
In this section, we describe the construction of a PCMT
under the general CMT framework of Section \ref{sec:prelims}. 
Assume all $N_j$'s are powers of 2. In Section \ref{sec:SEF_algo}, we remove this assumption. Let $\tI_j$ ($\tF_j$) be the information (frozen) index sets of the polar code used in layer $j$ of the PCMT.  We have $\vert \tI_j \vert = RN_j$ and $\vert \tF_j \vert = (1-R)N_j$. 
For convenience, we \deb{re-index} the row indices in FG $\tFG_j$ such that $\tI_j$ and $\tF_j$ are the indices $[1,RN_j]$ and $(RN_j,N_j]$, respectively.

For the PCMT, define intermediate coded symbols (which are used to form the PCMT) $\ttN_j[k][i]$, $\closedN{j}{l}$ where $\closedN{k}{\log N_j+1}$,  $\closedN{i}{N_j}$.  Index $k$ ($i$) is the column (row) number of the VN that the symbol $\ttN_j[k][i]$ corresponds to, in the FG  $\tFG_{N_j}$. 
In the general CMT framework, for $\closedN{j}{l}$, we have $\tttS[j] = \{\ttN_j[\log N_j + 1][i] \;\vert\; \closedN{i}{RN_j}\}$, $\tttP[j] = \{\ttN_j[\log N_j + 1][i] \;\vert\; \openleft{i}{RN_j}{N_j}\}$ and $\tttC[j] = \tttS[j] \cup  \tttP[j]$.

\subsubsection{\trim{Formation of PCMT symbols}}
The $\text{\tparity}()$ procedure for a PCMT is as follows: for the data symbols $\tttS[j]$, use a systematic polar encoder as described in Section \ref{sec:prelims} to find the parity symbols $\tttP[j]$, where VNs corresponding to $\tfrozen{[j]} =\{\ttN_{j}[1][i]\;\vert\; \openleft{i}{RN_j}{N_j}\}$ in $\tFG_{N_j}$ are set as zero chunks. The systematic encoder also provides the set of symbols $\tdropped[j] = \{\ttN_{j}[k][i] \;\vert\; \closedN{k}{\log N_j},\; \closedN{i}{N_j}\}$ which are dropped from the PCMT and are not included in $\tttC[j]$. \trim{However, before dropping, we use them to form the parent layer in the PCMT.
The $\text{\tparent}()$ procedure for a PCMT is as
follows.
Let $x \bmod {p} := (x)_p$ and let ${\fontfamily{qcr}\selectfont \text{Hash}}$ and ${\fontfamily{qcr}\selectfont \text{concat}}$ represent the hash and string concatenation functions, respectively. We have
} %
\vspace{-0.2cm}
\begingroup
\allowdisplaybreaks
\begin{align}\label{eqn:data_symbol_formation}
  \ttN_{j-1}&[\log N_{j-1}+1][i]\\&= {\fontfamily{qcr}\selectfont \text{concat}}( \{ {\fontfamily{qcr}\selectfont \text{Hash}}(\ttN_{j}[k][x]) \;\vert\; \closedN{k}{\log N_j+1},\nonumber\\&\; \closedN{x}{N_j},\; i = 1 + (x-1)_{RN_{j-1}} \}  ),\; \forall \closedN{i}{RN_{j-1}}\nonumber,
\end{align}
\endgroup
where $\tttS[j-1] = \{\ttN_{j-1}[\log N_{j-1} + 1][i] \;\vert \; \closedN{i}{RN_{j-1}}\}$.
For a PCMT, the root ($\tttS[0]$) has a size $t = N_1(\log N_1 + 1)$ hashes. In Fig. \ref{fig:PCMT_example}, the formation of $\ttN_{1}[3][2]$ is shown.

In the above $\text{\tparent}()$ procedure,  data symbols in $\tttS[j-1]$ are formed by taking the hashes of all the $N_j(\log N_j +1)$ intermediate coded symbols of layer $j$ (i.e., $\tdropped[j] \cup \tttC[j]$) and concatenating $q(\log N_j+1)$ hashes together according to Eqn. \eqref{eqn:data_symbol_formation}. 
The intuition behind taking the hashes of all the intermediate coded symbols is so that the symbols in $\tdropped[j]$ also get committed to the root (i.e., these symbols also have a Merkle proof). Although dropped, the symbols in $\tdropped[j]$ can be decoded back by a peeling decoder using the available (non-erased) symbols of $\tttC[j]$. Once decoded, they can be used to build small IC proofs using the degree 2 and 3 CNs in the polar FG $\tFG_{N_j}$.

\subsubsection{Merkle proof of PCMT symbols}
For the above PCMT construction, symbols in $\tttC[j]$ \trim{and} $\tdropped[j]$ have Merkle proofs. 
The Merkle proof of the symbols $\ttN_{j}[k][i]$, $\closedN{k}{\log N_j+1}, \; \closedN{i}{N_j}$ consists of a data symbol and parity symbol from each layer of the PCMT above $L_j$ similar to LCMT in \cite{AceD,SSskew-full}. Precisely, it is given by the following: 
\begin{align}
    &\text{\tproof}(\ttN_{j}[k][i]) = \{\ttN_{j'}[\log N_{j'} + 1][1 + (i-1)_{RN_{j'}}],\\& \ttN_{j'}[\log N_{j'} + 1][1 + RN_{j'} + (i-1)_{(1-R)N_{j'}}]\;\vert\; \closedN{j'}{j-1}\}\nonumber. 
\end{align}

\trim{The Merkle proof} for $\ttN_2[4][4]$ is shown in Fig. \ref{fig:PCMT_example}. 
The data symbols from each layer in $\text{\tproof}(\ttN_{j}[k][i])$ lie on the path of $\ttN_{j}[k][i]$ to the PCMT root; this path is used to
check the integrity of $\ttN_{j}[k][i]$ \trim{in a manner similar to  an LCMT \cite{CMT, AceD}.}

\subsubsection{Hash-aware peeling decoder and IC proofs}
The PCMT is decoded using a hash-aware peeling decoder similar to LCMT in \cite{CMT}. The $\tdecodelayer(L_j)$ procedure for the decoder is as follows. Its acts on the FG $\tFG_{N_j}$. It takes as inputs the frozen symbols $\tfrozen[j]$ and the non hidden symbols of $\tttC[j]$. Using a peeling decoder, it finds all symbols in $\tdropped[j] \cup \tttC[j]$. The hash of every decoded (peeled) symbol is matched with its hash provided by the parent layer $L_{j-1}$. If hashes do not match, an IC attack is detected. In this case, IC proof is generated using the degree 2 or 3 CN of the FG $\tFG_{N_j}$ as per the 
general CMT framework.

\vspace{-0.2cm}
\subsection{DA attacks on PCMT}
\vspace{-0.1cm}
Consider layer $L_j$, $\closedN{j}{l}$, of the PCMT. For a given information index set $\tI_j$, let $\tSS^{\tI_j}$ denote the set of all stopping sets in the FG $\tFG_{N_j}$ that
do not have any VNs corresponding to $\tfrozen[j]$.
The hash-aware peeling decoder fails to decode $L_j$ if coded symbols corresponding to a stopping set in $\tSS^{\tI_j}$ are erased. Since all the coded symbols except the rightmost column of $\tFG_{N_j}$ are dropped, the peeling decoder will fail if the adversary hides the leaf set of a stopping set in $\tSS^{\tI_j}$. 

To prevent a DA attack, light nodes randomly sample symbols from the PCMT base layer, i.e., $\tttC[l]$. Randomly sampling the base layer ensures that the non-dropped symbols of intermediate layer $L_j$, $\closedN{j}{l-1}$, i.e.,  $\tttC[j]$, are also randomly sampled via the Merkle proofs of the base layer samples similar to an LCMT in \cite{CMT}.
For subsequent analysis, we assume (WLOG) that the adversary conducts a DA attack on  the base layer of the PCMT. 
To find the adversary strategy that leads to the largest probability of failure when the light nodes use random  sampling, we use the following important property of stopping sets in polar FGs that was proved in \cite{Polar-SStree-TCOM}:
\begin{equation}\label{eqn:SStreeproperty}
    \min_{\tssingle \in \tSS^{\tI_l}}\vert\text{\tLS}(\tssingle)\vert = \min_{i \in \tI_l}f^{N_l}_{i}.
\end{equation}

Eqn. \eqref{eqn:SStreeproperty} implies that, when light nodes use random sampling, the best strategy for the adversary (to maximize the probability of failure) is to hide the smallest leaf set amongst all stopping trees with non frozen root. Thus, $\talpha = \underset{i \in \tI_l}{\min}f^{N_l}_{i}$.

\vspace{-0.2cm}
\section{Sampling-Efficient Freezing Algorithm}\label{sec:SEF_algo}
\vspace{-0.1cm}
For the best adversary strategy, $\talpha = \min_{i \in \tI_l}f^{N_l}_{i}$. 
Based
on this result,
a na\"\i ve frozen set selection method would be to select the indices of $RN_l$ VNs from the leftmost column of the FG $\tFG_l$ with the smallest stopping tree leaf set sizes $f^{N_l}_i$. Note that for this na\"\i ve frozen set selection, the polar code becomes equivalent to a Reed-Muller (RM) Code \cite{ReedMuller}. We call the na\"\i ve frozen set selection as \emph{Na\"\i ve-RM (NRM)} algorithm for which it
can be easily shown that $\talpha^{NRM} = \min\left(\ttT_{N_l};(1-R)N_l\right)$ \revision{(which is the $(1-R)N_l$-th smallest value of $\ttT_{N_l}$)}.

Next, we describe the Sampling-Efficient Freezing (SEF) algorithm and show that it results in a higher effective undecodable threshold compared to the NRM algorithm. Additionally, our algorithm allows for polar codes of any length and are not limited to powers of two. Assume that in this section, for all FG $\tFG_{\tNm}$, the rows in $\tFG_{\tNm}$ are indexed $1$ to $\tNm$ from top to bottom. The SEF algorithm is based on the following lemma. 

\vspace{-0.2cm}
\begin{lemma}\label{lemma:last_few_frozen}
Consider FG $\tFG_{\tNm}$ where $\tNm$  is a power of two. 
Let $\tFm$ and $\tIm$ be the frozen and information index sets. For a parameter $\tlast$, define the set of VNs $\tlastVN^{\tlast}_{\tNm}[k] = \{v_{ki}\; \vert\; \closed{i}{\tNm - \tlast +1}{\tNm}\}$.
If  $[\tNm - \tlast +1,\tNm] \subset \tFm$, then i) $\forall$  $\tssingle \in \tSS^{\tIm}$, $\tssingle$ does not have any VNs in $\tlastVN^{\tlast}_{\tNm}[\log \tNm + 1]$; ii) all VNs in $\{\tlastVN^{\tlast}_{\tNm}[k] \;\vert \; \closedN{k}{\log \tNm + 1} \}$ are zero chunks.
\vspace{-0.25cm}
\begin{proof}[Proof Idea]
Assuming that a stopping set $\tssingle \in \tSS^{\tIm}$ has a VN
from $\tlastVN^{\tlast}_{\tNm}[\log \tNm + 1]$, we prove, by incorporating the definition of stopping sets, that $\exists$ ${i}$, $\closed{i}{\tNm - \tlast+1}{\tNm}$ such that $v_{1i} \in \tssingle$. This is a contradiction since $\tssingle \in \tSS^{\tIm}$ and hence does not have VNs in $\tlastVN^{\tlast}_{\tNm}[1]$. Full proof can be found in Appendix \ref{appendix:proof_last_few_frozen}.
\end{proof}
\vspace{-0.35cm}

\end{lemma}

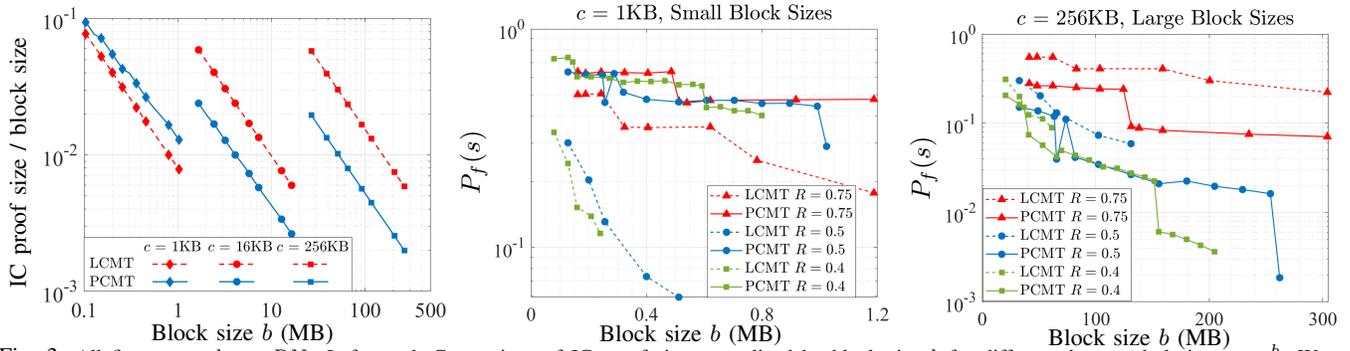
\begin{figure*}[t]
    \centering
    \begin{subfigure}{0.33\linewidth}
\begin{minipage}{0.99\linewidth}
    \begin{tikzpicture}
  \node (img)
  {\includegraphics[scale=0.2]{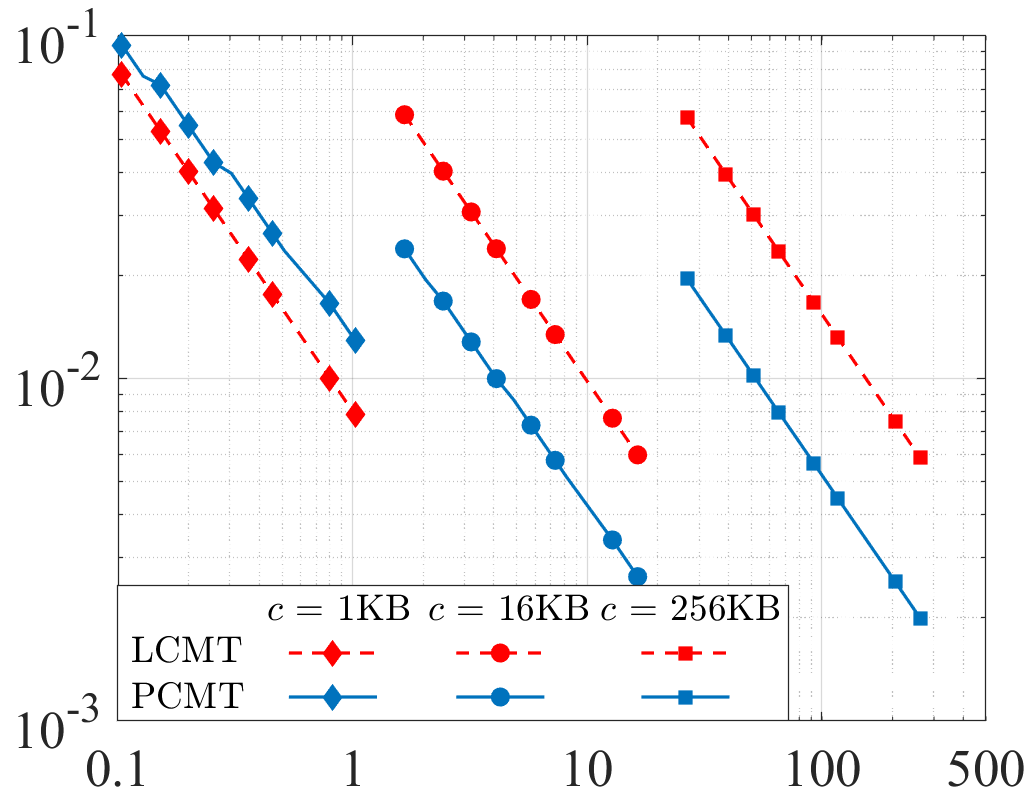}};
  \node[below=of img, node distance=0cm, yshift=1.2cm,font=\color{black}] {\small  Block size $b$ (MB)};
  \node[left=of img, node distance=0cm, rotate=90, anchor=center,yshift=-0.9cm,font=\color{black}] {\small IC proof size / block size};
 \end{tikzpicture}
 \end{minipage}
    \end{subfigure}%
    \begin{subfigure}{0.33\linewidth}
\begin{minipage}{0.99\linewidth}
\begin{tikzpicture}
  \node (img) {\includegraphics[scale=0.2]{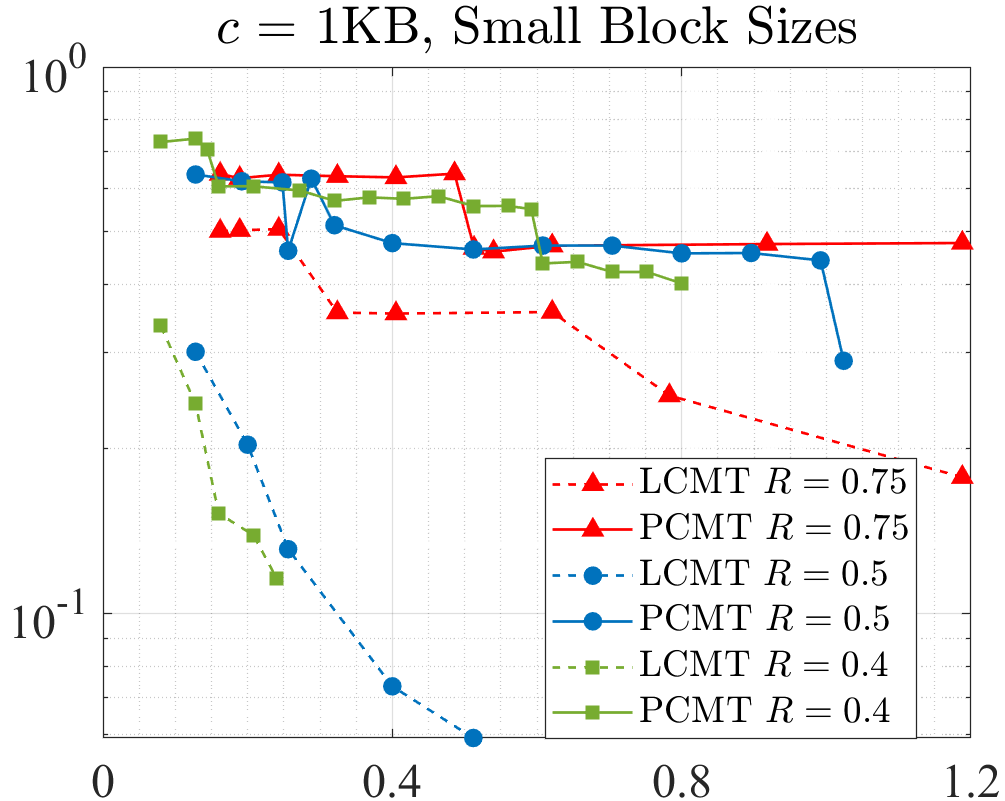}};
  \node[below=of img, node distance=0cm, yshift=1.2cm,font=\color{black}] {\small Block size $b$ (MB)};
  \node[left=of img, node distance=0cm, rotate=90, anchor=center,yshift=-0.9cm,font=\color{black}] {$P_f(s)$};
 \end{tikzpicture}
 \end{minipage}
    \end{subfigure}%
\begin{subfigure}{0.33\linewidth}
\begin{minipage}{0.99\linewidth}
\begin{tikzpicture}
  \node (img) {\includegraphics[scale=0.2]{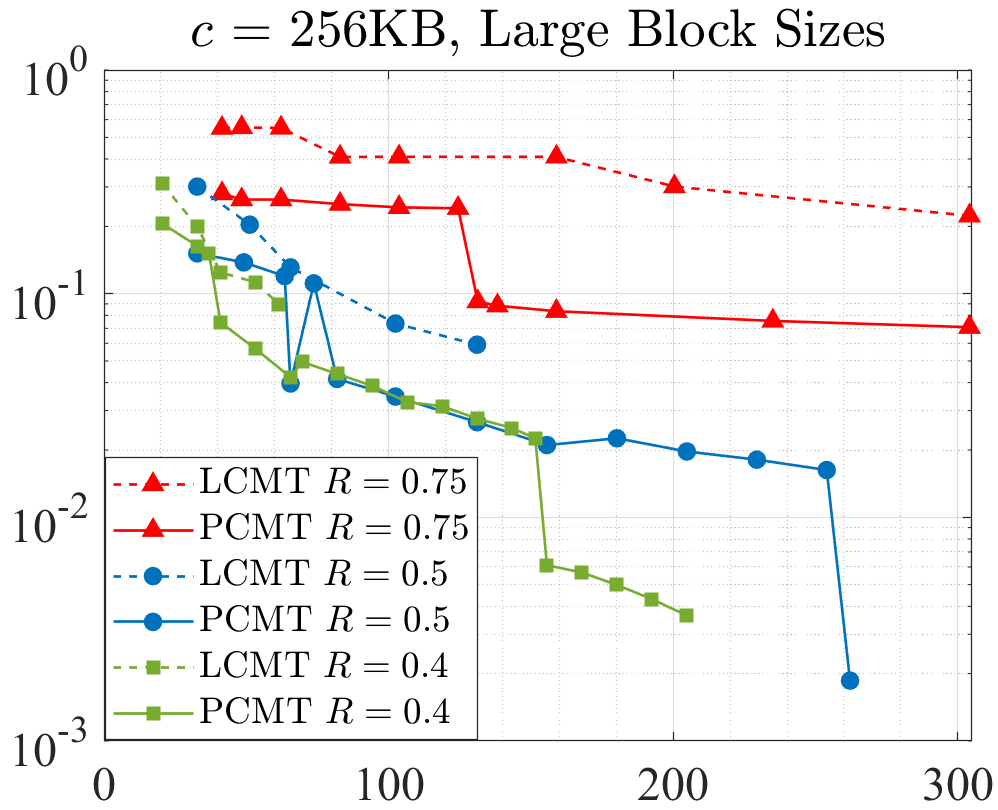}};
  \node[below=of img, node distance=0cm, yshift=1.25cm,font=\color{black}] {\small Block size $b$ (MB)};
  \node[left=of img, node distance=0cm, rotate=90, anchor=center,yshift=-0.9cm,font=\color{black}] {$P_f(s)$};
 \end{tikzpicture}
 \end{minipage}
 \vspace{-3pt}
    \end{subfigure}
     \vspace{-10pt}
    \caption{\footnotesize{\revision{All figures use $b = cRN_l$.} \deb{Left panel: Comparison of IC proof size normalized by block size $b$ for different data symbol size $c = \frac{b}{k}$. We use $(R,q,l) = (0.5,4,4)$ in the figure. For IC proof size of LCMT, we use the maximum CN degree $d_c$. For PCMT, the maximum CN degree $d_p = 3$; Middle and Right panel: $P_f(s)$ vs. blocksize $b$ for LCMT and PCMT. \revision{The two panels use  $(R,q,l) = (0.4,5,4)$, $(0.5,4,4)$, and $(0.75,4,3)$ %
    and a constant data symbol size $c$.} \revision{Sample size $s$ for PCMT and LCMT are selected such that total sample download size is $\frac{b}{3}$ and $\frac{b}{5}$ for the middle and right panels, respectively.}}}}
    \label{fig:IC_proof_P_f}
     \vspace{-18pt}
\end{figure*}

\vspace{-0.05cm}
Lemma \ref{lemma:last_few_frozen} states that, if the last $\tlast$ rows (from the bottom) in the leftmost column of $\tFG_{\tNm}$ are all frozen, then no stopping set in $\tSS^{\tIm}$ can have a VN from the last $\tlast$ rows in the rightmost column of $\tFG_{\tNm}$. Thus, for a frozen index set $\tFm$ such that
$[\tNm - \tlast +1,\tNm] \subset \tFm$, the light nodes do not need to sample the VNs in $\tlastVN^{\tlast}_{\tNm}[\log \tNm + 1]$.
We leverage the above property to improve the effective undecodable threshold of polar codes.
Additionally, since all the VNs in $\{\tlastVN^{\tlast}_{\tNm}[k] \;\vert\; \closedN{k}{\log \tNm+1}\}$, which are all the VNs 
in the last $\tlast$ rows of $\tFG_{\tNm}$, are zero chunks, these VNs and their associated edges can be removed from the FG. After this removal, we get the FG $\tFG_{\tNm - \tlast}$ of a polar code of length $\tNm - \tlast$. We use this property to design polar codes of lengths that are not powers of two. \trim{An example of FG $\tFG_{5}$
is shown in Fig. \ref{fig:Factor_graph} right panel.} Algorithm \ref{alg:SEF} provides the SEF algorithm to design the frozen index set $\tF$ of an $(N,K)$ polar code; $N$ is not necessarily a power of two.

\begin{algorithm}[t]
\caption{SEF Algorithm}\label{alg:SEF}
\begin{algorithmic}[1]
\State \textbf{Inputs:} $N$, $K$ \textbf{Output:} $\tFG_{N}$, $\tF$, $\ttT_N$ %
\State \textbf{Initialize:} $\tNm = 2^{\lceil \log N \rceil}$, FG $\tFG_{\tNm}$, $\tlast_1 = \tNm - N$, $i = N$.

\State $\tFG_{N}$ =  FG obtained by removing all VNs in $\{\tlastVN^{\tlast_1}_{\tNm}[k]\;\vert\; \closedN{k}{\log \tNm+1}\}$ and their connected edges from $\tFG_{\tNm}$ (also remove any CNs that have no connected edges)

\State $\ttT_{N}$ = $\ttT_{\tNm}$ with last $\tlast_1$ entries removed

\State $\tF =  \{e \;|\; \closedN{e}{N} ,\;  \ttT_{N}(e) < \min(\ttT_{N}; N - K)\}$

\While{$\vert \tF \vert < N- K$} 
\If{$i \not\in \tF$} 
$\tF  = \tF \cup i$ \textbf{end if};
 $i = i-1$
\EndIf
\EndWhile
\end{algorithmic}
\end{algorithm}

In the SEF algorithm, we first remove the VNs from the last $\tlast_1 = \tNm - N$ rows in FG $\tFG_{\tNm}$ (step 3). This gives us the FG $\tFG_{N}$ which has $N$ coded symbols and is used for the construction of the PCMT.   Then, in steps 4-7, we select the frozen index set $\tF$ of size $N-K$ for the $(N,K)$ polar code. Note that $\ttT_N$ (step 4) stores the stopping tree sizes of the VNs $v_{1i}$, $ \closedN{i}{N}$. For the selection of $\tF$, we first select all the indices $e$ in $[N]$ whose corresponding VNs $v_{1e}$ have their stopping tree sizes less than $\min(\ttT_N; N-K)$ (step 5). 
Then, the remaining indices in $\tF$ are selected as the VN indices from the bottom row of FG $\tFG_{N}$ that are not already present in $\tF$ (steps 6-7). We have the following lemma.

\vspace{-0.2cm}
\begin{lemma}\label{lemma:undecodable_threshold}
Let $\tlast_2$ be the largest $\tlast$ such that $[N - \tlast +1, N]  \subset \tF$ and let $\tI = [N]\setminus\tF$. For an $(N,K)$ polar code produced by the SEF algorithm, let the light nodes randomly sample among the top $N - \tlast_2$ VNs from the rightmost column of FG $\tFG_{N}$. For this sampling strategy, 
 the effective undecodable threshold is
 $\talpha^{SEF} = \frac{\min_{i \in \tI}\ttT_{N}(i)*N}{N - \tlast_2}$. As such, $P_f(s) = \left(1 - \frac{\talpha^{SEF}}{N}\right)^s$.
\end{lemma}
\begin{proof}[Proof Idea]
We prove the lemma by using Lemma \ref{lemma:last_few_frozen} and Eqn.
\ref{eqn:SStreeproperty}. 
The full proof can be found in Appendix \ref{appendix:undecodable_threshold}.
\end{proof}
\vspace{-0.2cm}

\trim{Note that due to step 5 of the SEF algorithm, $\min_{i \in \tI}\ttT_{N}(i) \geq \min(\ttT_N; N-K)$. Thus, the undecodable ratio of the SEF algorithm is always as big as the NRM algorithm. }
The FG $\tFG_{N}$ output by the SEF algorithm is used for constructing different layers of the PCMT. $\tFG_{N}$ has $\lceil \log N \rceil +1$ 
VN and $\lceil \log N \rceil$ CN columns each with $N$ VNs or CNs. Thus, for the PCMT construction described in Section \ref{sec:PCMT_construction}, we replace all instances of $\log N_j$
with $\lceil \log N_j \rceil$ (where $N_j$'s need not be powers of two). Rest of the construction remains the same.

\setlength{\extrarowheight}{3pt}
\setlength{\tabcolsep}{2.8pt}
  \begin{table}[t]
    \centering
\resizebox{!}{2.5cm}{
\begin{tabular}[b]{| c |c  c |c c |c c|}

\hline
& \multicolumn{2}{c|}{2D-RS} & \multicolumn{2}{c|}{LCMT}  & \multicolumn{2}{c|}{PCMT}  \\

& $\mathcal{T}_1$ & $\mathcal{T}_2$ & $\mathcal{T}_1$ & $\mathcal{T}_2$ & $\mathcal{T}_1$ & $\mathcal{T}_2$\\
\hline
Root size (KB) & 2.05 & 5.82 & 0.26 & 0.51 & 1.02 & 2.56\\
\hline
IC proof size (MB) & 5.80 & 16.40 & 1.54 & 1.54 & 0.53 & 0.54\\
\hline
Total sample download size (MB) & 10.76  & 12.30 & 43.01 & - & 33.80 & 80.10\\
\hline
 Decoding complexity & \multicolumn{2}{c|}{$O(N_l^{1.5})$} & \multicolumn{2}{c|}{$O(N_l)$}& \multicolumn{2}{c|}{$O(N_l\lceil\log N_l\rceil)$}\\
\hline
\end{tabular}
}  
    \caption{\footnotesize \revision{Comparison of various performance metrics for 2D-RS, LCMT and PCMT. The table uses $\mathcal{T}_1 = (k,R,q,l) = (512,0.5,4,8)$, $\mathcal{T}_2 = (4096,0.5,4,10)$, $c = 256$KB, and $b = cRN_l$. 
    Sample download size is calculated such that $P_f(s)$ is $0.01$. Due to complexity of finding $\talpha$ for LCMT, we do have a corresponding total sample download size value for $\mathcal{T}_2$.}  %
    }\label{figtable:numerical_comp}
  \end{table}

\vspace{-0.2cm}
\section{Simulations and Performance Comparison}\label{sec:sims}
\vspace{-0.15cm}
In this section, we demonstrate the benefits of a PCMT with
respect to the CMT metrics i)-iv) described in Section \ref{sec:prelims} \deb{when the size of the block $b$ is large}. We also compare the performance of a PCMT with an LCMT \revision{and 2D-RS codes}.  We denote the output size of the {\fontfamily{qcr}\selectfont \text{Hash}} function as $y$ and the size of the data chunk (symbol) in the base layer as $c$ \deb{where block size $b = ck = cRN_l$.}  We use $y = 256$ bits.  All PCMTs are built using SEF Polar codes. All LCMTs are built using LDPC codes constructed using the PEG algorithm \cite{PEG} where the degree of all VNs is set to 3. For the PEG LDPC codes, the undecodable threshold $\talpha$ is evaluated by solving an Integer Linear Program (ILP) as described in \cite{ILPsearch} and is computationally infeasible for larger code lengths.
\deb{In contrast, SEF Polar codes have an easily computable $\talpha$  using Lemma \ref{lemma:undecodable_threshold}}.
 \deb{Due to complexity issues, we compute $\talpha$ for LDPC  codes up to a feasible code length (and, thus, up to a feasible block size).}
\deb{From the description of the PCMT provided in Section \ref{sec:PCMT_construction}, derivation of its root size, IC proof size, and single sample download size (base layer sample and associated Merkle proof) is straight forward and is provided in Appendix \ref{appendix:table_calculation} \revision{(which we use to generate the plots in Fig. \ref{fig:IC_proof_P_f})}.}

\vspace{-0.01cm}
\lev{
\trim{In Fig. \ref{fig:IC_proof_P_f} left panel, we plot the IC proof size vs. block size $b$ for an LCMT and a PCMT and different data symbol sizes $c$.}
We see that for $c = 256$ and $16$KB (large block sizes), the IC proof size is smaller for a PCMT compared to an LCMT and gets bigger than an LCMT for $c=1$KB (small block sizes).}

\vspace{-0.2cm}
\revision{
\begin{remark} 
We note that for a PCMT and an LCMT with the same CMT $\mathcal{T}$ parameters, the PCMT incurs an asymptotic penalty of $O(\log b)$ in the IC-proof size over the LCMT due to collecting the hashes of VNs in all the columns of the FG (can be seen from the expressions in Appendix \ref{appendix:table_calculation}). However for practical block sizes of interest, the IC-proof size of a PCMT can be significantly lower compared to an LCMT, as shown in Fig. \ref{fig:IC_proof_P_f}, due to the low CN degree in the FG of polar codes.
\end{remark}
}
\vspace{-0.2cm}

\deb{In Fig. \ref{fig:IC_proof_P_f} middle and right panels, we compare the probability of failure $P_f(s)$ to detect a DA attack conducted on the base layer. We compare $P_f(s)$ for large and short block sizes in the middle and right panels, respectively.}
\lev{From Fig. \ref{fig:IC_proof_P_f}, we see that PCMT has a worse probability of failure compared to an LCMT for small block sizes. However for large block sizes, PCMT always has a lower probability of failure compared to an LCMT across all rates $R$ and block sizes $b$ thanks to a higher undecodable ratio for the SEF Polar codes and a negligible penalty in the single sample download size.%
}

\revision{In Table \ref{figtable:numerical_comp}, we provide additional comparison of various performance metrics for 2D-RS, LCMT, and PCMT. We can see that PCMT outperforms LCMT w.r.t. IC-proof size and total sample download size with small increase in root size and decoding complexity. While PCMT has a $O(\lceil \log N_l \rceil)$ factor greater decoding complexity than an LCMT, the decoding complexity is smaller than for 2D-RS codes which is $O(N_l^{1.5})$ \cite{CMT}.
At the same time, PCMT also has a lower IC proof size and root size compared to 2D-RS codes while having a higher sample download size. Going from $\mathcal{T}_1$ to $\mathcal{T}_2$, the IC-proof size for 2D-RS codes increases 3 fold while the IC-proof size remains almost constant for LCMT and PCMT. Note that for 2D-RS codes, the IC proof size, decoding complexity, and header size do not scale well as the block size increases \cite{CMT}.}

 Overall, when the size of the transaction block $b$ is large, a PCMT built using SEF Polar codes has good performance w.r.t metrics i)-iv) described in Section~\ref{sec:prelims} and offers a new trade-off in these metrics compared to LCMT and 2D-RS codes.

\appendices

\vspace{-0.1cm}
\section{Proof of Lemma \ref{lemma:last_few_frozen}\label{appendix:proof_last_few_frozen}}

First of all, it is easy to see that when all the VNs in $\tlastVN^{\tlast}_{\tNm}[1]$ (i.e., the VNs in the last $\tlast$ rows from the leftmost column of FG $\tFG_{\tNm}$) are frozen or set to zero chunks, all the VNs in the last $\tlast$ rows from all the columns will be zero chunks. This is because every CN in row $i$ of the FG is either connected to VNs in the same row $i$ or to VNs from a row below $i$ (i.e., having a row index greater than $i$) in the FG. This proves the second claim of the lemma. 

For the first claim, let $\tssingle \in \tSS^{\tIm}$ and let $\tFG^{\tssingle}_{\tNm}$ be the induced subgraph of $\tFG_{\tNm}$ corresponding to the set of VNs $\tssingle$. From the definition of $\tSS^{\tIm}$,  $\tssingle$ does not have any frozen VNs from the leftmost column of the FG $\tFG_{\tNm}$, i.e., $\tssingle$ does not have any VNs in the set $\{v_{1i} \;\vert \; i \in \tFm\}$. Since  $[\tNm - \tlast +1, \tNm] \subset \tFm$, $\tssingle$ does not have any VNs in $\tlastVN^{\tlast}_{\tNm}[1]$. 

We prove the first claim of the lemma by contradiction. Assume that $\tssingle$ has a VN from $\tlastVN^{\tlast}_{\tNm}[\log \tNm + 1]$. In particular, assume that $v_{(n+1)i_1} \in \tssingle$, where $n = \log \tNm$ and $\closed{i_1}{\tNm -\tlast+1}{\tNm}$.
Now, by the property of stopping sets, $c_{ni_1} \in \tFG^{\tssingle}_{\tNm}$. Now, to satisfy the stopping set property, either $v_{ni_1} \in  \tssingle$ or $v_{ni_2} \in \tssingle$ where
$i_1 < i_2 \leq \tNm$ and $v_{ni_2}$ and $c_{ni_1}$ are connected in $\tFG_{\tNm}$. Thus, for the column number $n$, we have at least one index ${i}$, $\closed{i}{\tNm - \tlast +1 }{\tNm}$  such that $v_{ni} \in \tssingle$. Proceeding in a similar manner as above, for the column number $(n-1)$, we have at least one index ${i}$, $\closed{i}{\tNm - \tlast +1 }{\tNm}$  such that $v_{(n-1)i} \in \tssingle$. Repeating the same process until we reach the leftmost column, we can find at least one index ${i}$, $\closed{i}{\tNm - \tlast +1 }{\tNm}$  such that $v_{1i} \in \tssingle$. However, this is a contradiction of the fact that $\tssingle$ does not have any VNs in set $\tlastVN^{\tlast}_{\tNm}[1] = \{v_{1i} \;|\; \closed{i}{\tNm - \tlast +1 }{\tNm}\}$. 

\section{Proof of Lemma \ref{lemma:undecodable_threshold}\label{appendix:undecodable_threshold}}
 The SEF algorithm produces a $(N,K)$ polar code with a FG $\tFG_N$ where the bottom $\tlast_2$ VNs from the leftmost column of $\tFG_N$ are frozen. Moreover, $\tFG_N$ is obtained from freezing (and hence removing) the last $\tNm - N$ rows of $\tFG_{\tNm}$, where $\tNm = 2^{\lceil \log N \rceil}$.  Note that $\tF$ is the output of the SEF algorithm and $\tI = [N]\setminus\tF$.
 Define $\tFm = \tF \cup [N + 1, \tNm]$, $\tIm = [\tNm]\setminus \tFm$. Clearly, $\tIm$ and $\tI$ are the same sets. Thus, the $(N,K)$ polar code can be seen as a code defined on the FG $\tFG_{\tNm}$ with frozen index set $\tFm$, and information index set $\tI$, where only the top $N$ VNs from the rightmost column of $\tFG_{\tNm}$ are the coded symbols. 
 
 Due to Lemma \ref{lemma:last_few_frozen}, VNs in the last $\tlast_2 + \tNm - N$ rows in the rightmost column of $\tFG_{\tNm}$ are not part of any stopping set in $\tSS^{\tI}$. This implies that VNs in the last $\tlast_2$ rows in the rightmost column of $\tFG_N$ are not part of any stopping set in $\tSS^{\tI}$. 
 Thus in the FG $\tFG_{N}$, the light nodes do not need to sample the VNs in the last $\tlast_2$ rows in the rightmost column of $\tFG_N$ and only randomly sample the top $N - \tlast_2$ VNs. Moreover, %
 from Eqn. \eqref{eqn:SStreeproperty} (applied on FG $\tFG_{\tNm}$), the smallest leaf set size of all stopping sets in $\tSS^{\tI}$ is given by
 \begin{align*}
     \min_{\tssingle \in \tSS^{\tI}}\vert\text{\tLS}(\tssingle)\vert = \min_{i \in \tI}f^{\tNm}_{i}
     = \min_{i \in \tI}\ttT_{\tNm}(i)
     = \min_{i \in \tI}\ttT_{N}(i).
 \end{align*}
 Hence, the probability of failure $P_f(s) = (1 - \frac{\min_{i \in \tI}\ttT_{N}(i)}{N - \tlast_2})^s$ resulting in an effective undecodable threshold of $\talpha^{SEF} = \frac{\min_{i \in \tI}\ttT_{N}(i)*N}{N - \tlast_2}$.

\setlength{\extrarowheight}{3pt}
\setlength{\tabcolsep}{2.8pt}
  \begin{table*}[t]
    \centering
\resizebox{!}{3.5cm}{
\begin{tabular}[b]{| L | M | N | O|}

\hline
& 2D-RS & LCMT  & PCMT  \\
\hline
Root size & $2y\lceil\sqrt{N_l}\rceil$ &$y N_1$ & $ yN_1(\lceil \log N_1\rceil+1)$\\
\hline
Single sample download size & $\frac{b}{k} + y\lceil \log \sqrt{N_l} \rceil$ & $\frac{b}{k} + y(2q-1)(l-1)$ & $\frac{b}{k} + y(2q-1)(l-1) + 2qy\sum_{j=1}^{l-1}\lceil \log N_j\rceil$\\
\hline
IC proof size &  $(\frac{b}{k} + y\lceil \log \sqrt{N_l} \rceil)\lceil \sqrt{k}\rceil$& $\frac{(d_c-1)b}{k} + d_cy(q-1)(l-1)$ & $ \frac{(d_p-1)b}{k} + d_py(q-1)(l-1) + d_pqy\sum_{j=1}^{l-1}\lceil \log N_j\rceil$\\
\hline
Decoding complexity & $O(N_l^{1.5})$ & $O(N_l)$& $O(N_l\lceil\log N_l\rceil)$\\
\hline
$\talpha$ & Analytical expression in \cite{dataAvailOrg} &NP-hard to compute & Lemma \ref{lemma:undecodable_threshold}\\
\hline
\end{tabular}
}  
    \caption{\footnotesize Comparison of various performance metrics of 2D-RS codes, an LCMT, and a PCMT. The LCMT and PCMT have the same $(k,R,q,l)$ parameters. The maximum degree of CNs  in the LDPC codes and polar FG are $d_c$ and $d_p = 3$, respectively. The size of the block is $b$. 2D-RS has $k$ data symbols and $\lceil \log \sqrt{N_l}\rceil$ layers in the Merkle tree where $N_l = \frac{k}{R}$.  For an LCMT, the number of layers $l$ can be calculated such that the root size is some fixed constant $t$. The same $l$ is used for the PCMT. 
    }\label{figtable:Performance_comp}
    \vspace{-15pt}
  \end{table*}
\section{Performance Analysis and  Comparison}\label{appendix:table_calculation}

Comparison of various performance metrics of an LCMT, a PCMT and 2D-RS codes is provided in Table I. The metrics for 2D-RS codes are calculated as described in \cite{dataAvailOrg}.
Detailed derivation of the root size, single sample download size, and IC proof size  for an LCMT and PCMT in Table I is as follows.

\subsubsection{Root size} For an LCMT, the root consists of the hashes of all the coded symbols in $L_1$. Hence, the root consists of $N_1$ hashes and, thus, has a size of $yN_1$. 
For a PCMT, the root consists of hashes of all the coded and dropped symbols in $L_1$, i.e., the hashes of all the VNs in FG of the polar code used in $L_1$. Hence, the root consists of $N_1(\lceil\log N_1 \rceil+1)$ hashes and, thus, has a size of $yN_1(\lceil\log N_1 \rceil+1)$.  
  
\subsubsection{Single sample download size} For an LCMT, as described in \cite{CMT, AceD, SSskew-full}, each sample request consists of a base layer symbol of the CMT and the Merkle proof of the base layer symbol. Moreover, the Merkle proof of the base layer symbol consists of a data symbol and a parity symbol from each layer above the base layer (i.e., layers $L_j$, $\closedN{j}{l-1}$
). The Merkle proof satisfies the property that the data symbol in proof from layer $L_j$ consists of the hash of the data symbol in proof from layer $L_{j+1}$, $\closedN{j}{l-1}$. Thus, of the $q$ hashes present in the data symbol of the Merkle proof from layer $L_j$, $\closedN{j}{l-1}$, the hash corresponding to the data symbol
of the Merkle proof from layer $L_{j + 1}$ is not communicated in the Merkle proof (and  it can be calculated by taking a hash of the data symbol of the Merkle proof from layer $L_{j + 1}$). Thus,
there are only $(q-1)$ hashes from each layer $L_j$, $\closedN{j}{l-1}$ for the data part in the Merkle proofs.  Thus, the size of the Merkle proof of a base layer symbol is $y(2q-1)(l-1)$ (since the size of each parity symbol is $yq$). Finally, the overall download size for a single sample request is $\frac{b}{k} + y(2q-1)(l-1)$, where $\frac{b}{k}$ is the size of the base layer symbol. 

For a PCMT, similar to an LCMT above,  a sample request consists of a base layer symbol of the PCMT and the Merkle proof of the base layer symbol. The Merkle proof of the base layer symbol  in this case again consists of a data symbol and a parity symbol from each layer above the base layer. Note that, here, the data and the parity symbols in layer $L_j$ are the non-dropped symbols of the polar FG $\tFG_{N_j}$ corresponding to the information $\tI_j$ and frozen set $\tF_j$ indices, respectively. In contrast to an LCMT, in a PCMT, each data symbol of  layer $L_j$, $\closedN{j}{l-1}$ consists of $q_j = q (\lceil \log N_j \rceil +1)$ hashes. 

The Merkle proof in a PCMT also satisfies the property  that the data symbol in proof from layer $L_j$ consists of the hash of the data symbol in proof from layer $L_{j+1}$, $\closedN{j}{l-1}$. Thus, only $(q_j-1)$ hashes from  layers $L_j$,  $\closedN{j}{l-1}$ are present in the data part of the Merkle proofs. Thus, for a PCMT,  the size of the Merkle proof of a base layer symbol is
\begin{align*}
    \sum_{j= 1}^{l-1}y(2q_j-1) &= \sum_{j=1}^{l-1}y(2(q (\lceil \log N_j \rceil +1))-1)\\& = y (2q-1)(l-1) + 2qy\sum_{j=1}^{l-1}\lceil \log N_j \rceil
\end{align*}

Finally, the overall download size for a single sample request is $\frac{b}{k} + y(2q-1)(l-1) + 2qy\sum_{j=1}^{l-1}\lceil \log N_j \rceil$, where $\frac{b}{k}$ is the size of the base layer symbol in a PCMT. 
Note the additional penalty factor of $2qy\sum_{j=1}^{l-1}\lceil \log N_j \rceil$ in the single sample download size
for a PCMT compared to an LCMT. However, for large block sizes when $\frac{b}{k}$ is large compared to $y$, the penalty is small.

\subsubsection{IC proof size}
As described in Section \ref{sec:prelims}, the IC proof for a failed parity check equation with $d$ symbols consists of $d-1$ symbols and the Merkle proofs of the $d$ symbols. Note that the proof size is largest for a failed parity check equation in the base layer. Thus, we provide the IC proof size when the $d$ symbols are base layer symbols.   Also note that, in IC proofs, the Merkle proof of a symbol only consists on the data symbols from each layer above the base layer (the parity symbols included in the Merkle proofs of the light node samples are only to get additional samples of the intermediate layers) \cite{CMT,AceD}. Thus, in an LCMT, for a failed parity check equation with $d$ symbols, the size of the IC proof is $\frac{(d-1)b}{k}  + dy(q-1)(l-1)$, where the $(q-1)$ term arises due to the same reason as explained in the single sample download size calculation. Hence, when the maximum CN degree of the LDPC code is $d_c$, the IC proof size becomes  $\frac{(d_c-1)b}{k}  + d_cy(q-1)(l-1)$.

For a PCMT, the IC proof for a failed parity check equation with $d$ symbols again consists of $d-1$ symbols and the Merkle proofs of the $d$ symbols. Note that the symbols here can be both the dropped or non-dropped symbols of the polar FG. Also, each data symbol of  layer $L_j$, $\closedN{j}{l-1}$ in a PCMT consists of $q_j = q (\lceil \log N_j \rceil +1)$ hashes. Thus, for a failed parity check equation with $d$ symbols, the size of the IC proof in a PCMT is 
\begin{align*}
    \frac{(d-1)b}{k}  &+ \sum_{j=1}^{l-1}dy(q_j -1)\\ &= \frac{(d-1)b}{k} + \sum_{j=1}^{l-1}dy((q (\lceil \log N_j \rceil +1)) -1)\\
    &= \frac{(d-1)b}{k} + dy(q-1)(l-1) + dyq\sum_{j=1}^{l-1}\lceil \log N_j \rceil.
\end{align*}

For a PCMT, the maximum CN degree for a CN in the FG of the polar code is $d_p = 3$. In this case, the IC proof size becomes  $\frac{(d_p-1)b}{k} + d_py(q-1)(l-1) + d_pyq\sum_{j=1}^{l-1}\lceil \log N_j \rceil$.

\subsection{Undecodable Threshold}

The undecodable threshold $\talpha$ for an LCMT with the LDPC codes described in Section \ref{sec:sims} and a PCMT with SEF Polar codes is shown in Fig. \ref{fig:P_f_undecodable_threshold}.
From the figure, we see that a PCMT has a higher $\talpha$ compared to an LCMT for different values of rates $R$ and code lengths $N$. Hence for large block size $b$, when the penalty in the single sample download size for a PCMT is small, they have a lower probability of failure compared to an LCMT for the same total sample download size. This result is shown in Fig. \ref{fig:IC_proof_P_f} right panel.
\begin{figure}[t]
    \centering
\begin{minipage}{0.99\linewidth}
 \centering
\begin{tikzpicture}
  \node (img)
  {\includegraphics[scale=0.3]{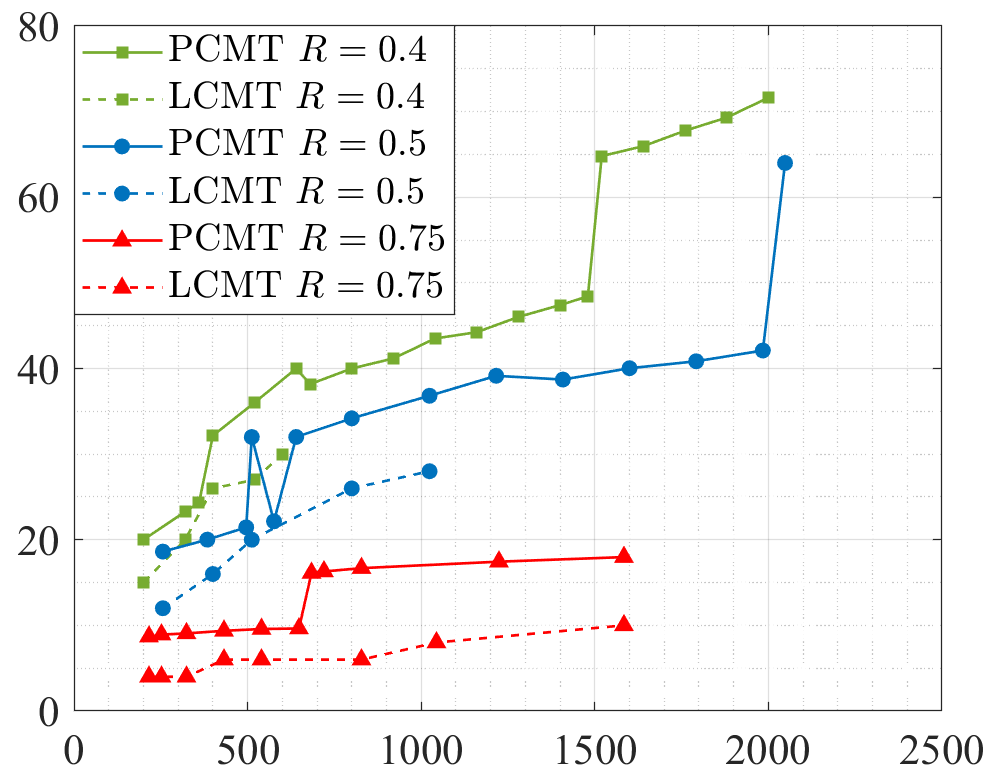}};
  \node[below=of img, node distance=0cm, yshift=1.2cm,font=\color{black}] {$N$};
  \node[left=of img, node distance=0cm, rotate=90, anchor=center,yshift=-1cm,font=\color{black}] {$\talpha$};
 \end{tikzpicture}
 \end{minipage}
    \vspace{-5pt}
    \caption{\footnotesize Left Panel: Undecodable threshold $\talpha$ vs $N$ for PEG LDPC codes (with VN degree 3) and SEF Polar codes. For PEG LDPC codes, we only list $\talpha$ that we were able to find by solving the ILP in \cite{ILPsearch}. 
      }
    \label{fig:P_f_undecodable_threshold}
\end{figure}

\section{Systematic encoding of Polar codes }\label{appendix:PDE}

 Systematic encoding of an $(\tNm,\tkm)$ polar codes can be performed using the method described in \cite{SystematicPOlarCodestwostatge}. It operates on the FG $\tFG_{\tNm}$ of the code. Let $\tIm \subset [\tNm]$ and $\tFm = [\tNm] \setminus \tIm$ be the index sets corresponding to the data and frozen symbols of the $(\tNm,\tkm)$ polar code, respectively, where $\vert \tIm \vert = \tkm$. Also, let $\tNm = 2^n$. The systematic encoding is performed on the FG $\tFG_{\tNm}$ by i) placing the $\tkm$ data symbols at the VNs $\{v_{(n+1)i} \vert i \in \tIm\}$ (in the rightmost column) and setting the VNs at $\{v_{1i} \vert i \in \tFm\}$ (in leftmost column) to zero symbols; ii) determining the rest of the VNs using the check constraints of the FG $\tFG_{\tNm}$ in a two stage reverse and forward encoding on the FG \cite{SystematicPOlarCodestwostatge}. The coded symbols are the VNs $v_{(n+1)i}, \closedN{i}{\tNm},$ and  are by design systematic.
Systematic encoding can also be performed using a peeling decoder as the encoder as described next.

In \cite{SystematicPOlarCodes}, author proposed to perform systematic encoding of polar codes by using a successive cancellation decoder on the code FG as the encoder. Inspired by this idea, we show that systematic encoding of the polar codes can also be performed using a peeling decoder as the encoder. The encoder, which we call a peeling encoder for polar codes (PEPC), works as follows.
 Consider the data and frozen index sets $\tIm \subset [\tNm]$ and $\tFm = [\tNm] \setminus \tIm$ of the polar code with FG $\tFG_{\tNm}$. Place the data symbols at the VNs $\{v_{(n+1)i} \vert i \in \tIm\}$ (in the rightmost column) and set the VNs at $\{v_{1i} \vert i \in \tFm\}$ (in leftmost column) to zero symbols. Use a peeling decoder to find the remaining VNs of the FG. The coded symbols $c_i = v_{(n+1)i}, \closedN{i}{\tNm},$ are by design systematic. We have the following lemma corresponding to a PEPC.

\begin{lemma}\label{lemma:PDE_sucessfull}
Systematic encoding of polar codes using a PEPC always results in a valid codeword. In other words, the peeling decoder never encounters a decoding failure when used for encoding. 
\end{lemma}

We prove Lemma \ref{lemma:PDE_sucessfull} by proving the following important property of stopping sets in the FG of polar codes produced by the SEF algorithm. 
Note that the property holds true for regular polar FG $\tFG_{\tNm}$ where $\tNm$ is a power of 2. Additionally, the property also holds true for the FG $\tFG_{N}$ obtained by removing VNs from the last few rows of FG of the form $\tFG_{\tNm}$, where $\tNm$ is a power of 2. Thus, let $n = \lceil \log N \rceil$. The proof provides insights on important properties of stopping sets in the FG of polar codes. To our best knowledge, we have not seen the following result before in literature and, hence, it may be of independent interest. 

\begin{lemma}\label{lemma:full_row}
 Consider a polar FG $\tFG_{N}$ produced by the SEF algorithm (this encompasses  polar factor graphs $\tFG_{\tNm}$, where $\tNm$ is a power of 2). Every stopping set of $\tFG_{N}$ has a full row of VNs i.e., it has all VNs in the set $\{v_{ki} \;\vert \; \closedN{k}{\;\lceil \log N \rceil + 1\;}  \}$ for some $ \closedN{i}{N}$. 
\end{lemma}
\begin{proof}
Let $\tssingle$ be a stopping set of $\tFG_{N}$. Let $\tFG^{\tssingle}_N$ be the induced subgraph of $\tFG_{N}$ corresponding to the set of VNs $\tssingle$.  Observe that the FG $\tFG_N$ has two types of edges (see Fig. \ref{fig:Factor_graph} for an example): horizontal edges and slanted edges (which involves a connection between a degree 3 VN and a degree 3 CN). We consider two cases: i) $\tFG^{\tssingle}_N$ does not have any slanted edges;
ii) $\tFG^{\tssingle}_N$  has at least one slanted edge.

For case i), it is easy to see that the stopping set $\tssingle$ must include a full row of VNs. For case ii), note that it implies $\tFG^{\tssingle}_N$  has at least one slanted edge. This implies that $\tssingle$ has at least one VN of degree 3. In this situation, define the set 
$\Delta_{\tssingle} = \{(i,k) \;\vert\; \closedN{i}{N},\; \closedN{k}{n}, v_{ki} \in \tssingle, \text{degree of } v_{ki} = 3 \}$. Let $i_{\max} = \max(\{i \vert (i,k) \in \Delta_{\tssingle} \text{ for some } k,\closedN{k}{n}\})$.
$\Delta_{\tssingle}$ contains the indices of all the degree 3 VNs of $\tssingle$ and $i_{\max}$ denotes the largest row index such that $\tssingle$ has a degree 3 VN from that row. Note that  $\Delta_{\tssingle}$ is non empty because we are considering case ii). Now, we claim that $\tssingle$ has all the VNs in the row $i_{\max}$, i.e., $\tssingle$  contains all the VNs in $\{v_{ki_{\max}}\;\vert\;\closedN{k}{n+1}\}$.  To see why this is true, let $\thickbar{k}$, $\closedN{\thickbar{k}}{n}$, be such that $(i_{\max}, \thickbar{k}) \in \Delta_{\tssingle}$. This means that $v_{\thickbar{k}i_{\max}} \in \tssingle$. Now by the definition of a stopping set, the CNs to the right and left of $v_{\thickbar{k}i_{\max}}$ must belong to the induced subgraph graph of the stopping set. In other words,  $c_{(\thickbar{k}-1)i_{\max}} \in \tFG^{\tssingle}_N$ and  $c_{\thickbar{k}i_{\max}} \in \tFG^{\tssingle}_N$ (unless $v_{\thickbar{k}i_{\max}}$ is the rightmost or the left most VN in which case we will have only one CN neighbour). Now, to satisfy the stopping set property, for both these CNs, their corresponding VNs to their left and right in the same row  $i_{\max}$ must belong to the stopping set $\tssingle$ (i.e., $v_{(\thickbar{k}-1)i_{\max}} \in \tssingle$, $v_{(\thickbar{k}+1)i_{\max}} \in \tssingle$). If not, then to satisfy the stopping set property, the CN must be connected to a VN $v_{ki}\in \tssingle$ by a slanted edge. Note that a slanted edge connects a CN to a degree 3 VN in lower row. In other words,  a slanted edge connects a CN from row $i_{\max}$ to a degree 3 VN in a row with index greater than $i_{\max}$. This condition violates the definition of $i_{\max}$. Thus, $v_{(\thickbar{k}-1)i_{\max}} \in \tssingle$ and $v_{(\thickbar{k}+1)i_{\max}} \in \tssingle$. 
Now, considering $v_{(\thickbar{k}-1)i_{\max}}$ and $v_{(\thickbar{k}+1)i_{\max}}$ as the starting VN (similar to $v_{\thickbar{k}i_{\max}}$), we can apply the above logic to show that $v_{(\thickbar{k}-2)i_{\max}} \in \tssingle$ and $v_{(\thickbar{k}+2)i_{\max}}\in \tssingle$. Repeatedly applying the same argument, we can show that all the VNs in $\{v_{ki_{\max}}\;\vert\;\closedN{k}{n+1}\}$ belong to $\tssingle$, where $n = \lceil \log N \rceil$.

\end{proof}

\vspace{-0.2cm}
According to the above lemma, every stopping set has at least one full row of VNs, i.e., it has VNs from all columns along a single row. We use this property to prove Lemma \ref{lemma:PDE_sucessfull}. \\[-2mm]

\noindent
\emph{Proof of Lemma \ref{lemma:PDE_sucessfull}}:\\
Let $\mathcal{V}$ be the set of all VNs in the polar FG $\tFG_N$. Also, let $\mathcal{\overline{V}}_{e} = \{v_{(n+1)i}\;\vert\;i\in \tI\} \cup \{v_{1i}\;\vert\;i\in \tF\}$ and $\mathcal{V}_e = \mathcal{V} \setminus \mathcal{\overline{V}}_{e}$. VNs in $\mathcal{V}_e$ are uninitialized or initially erased by the PEPC and are determined using a peeling decoder (for the purposes of encoding). 
Thus, the PEPC can fail if $\mathcal{V}_e$ contains a stopping set of $\tFG_N$. However, since $\tI \cup \tF = [N]$ (i.e., they form a partition of all the row indices), for all $\closedN{i}{N}$ either $v_{(n+1)i} \in \mathcal{\overline{V}}_{e}$ or $v_{1i} \in \mathcal{\overline{V}}_{e}$. Thus, for all $\closedN{i}{N}$, $v_{(n+1)i}$ and $v_{1i}$ both cannot simultaneously belong to $\mathcal{V}_e$. Hence, $\mathcal{V}_e$ cannot contain a full row of VNs. Using Lemma \ref{lemma:full_row}, we conclude that $\mathcal{V}_e$ cannot contain a stopping set. Thus, the PEPC will always be successful and will result in a valid codeword.

\end{document}